\documentclass[english]{article}

\usepackage{units}
\usepackage{amsmath}
\usepackage{amsthm}
\usepackage{amssymb}
\usepackage{fullpage}

\usepackage[T1]{fontenc}
\usepackage[utf8]{inputenc}
\usepackage{lmodern}    
\usepackage{microtype}
\usepackage{hyperref}

\usepackage{graphicx}

\theoremstyle{plain}
\newtheorem{thm}{\protect\theoremname}[section]
  \theoremstyle{definition}
  \newtheorem{defn}[thm]{\protect\definitionname}
\newenvironment{lyxlist}[1]
{\begin{list}{}
{\settowidth{\labelwidth}{#1}
 \setlength{\leftmargin}{\labelwidth}
 \addtolength{\leftmargin}{\labelsep}
 }}
{\end{list}}
  \theoremstyle{definition}
  \newtheorem{example}[thm]{\protect\examplename}
  \theoremstyle{plain}
  \newtheorem{lem}[thm]{\protect\lemmaname}
  \theoremstyle{remark}
  \newtheorem{rem}[thm]{\protect\remarkname}
  \theoremstyle{plain}
  \newtheorem{cor}[thm]{\protect\corollaryname}
  \theoremstyle{remark}
  \newtheorem*{rem*}{\protect\remarkname}
  \theoremstyle{plain}
  \newtheorem*{thm*}{\protect\theoremname}

\usepackage{babel}
  \providecommand{\corollaryname}{Corollary}
  \providecommand{\definitionname}{Definition}
  \providecommand{\examplename}{Example}
  \providecommand{\lemmaname}{Lemma}
  \providecommand{\remarkname}{Remark}
  \providecommand{\theoremname}{Theorem}
  \providecommand{\theoremname}{Theorem}

\global\long\def\ETH{\mathsf{ETH}}
\global\long\def\sharpETH{\mathsf{\#ETH}}

\global\long\def\rETH{\mathsf{rETH}}
\global\long\def\sharpP{\mathsf{\#P}}
\global\long\def\pmult#1{\mathbf{#1}}
\global\long\def\coeff#1{\mathsf{Coeff}(#1)}
\global\long\def\leqSERF{\leq_{\mathit{serf}}^{T}}
\global\long\def\eval#1{\mathsf{Eval}(#1)}
\global\long\def\A{\mathsf{A}}
\global\long\def\B{\mathsf{B}}
\global\long\def\evalfix#1#2{\mathsf{Eval}_{#1}(#2)}
\global\long\def\perm{\mathrm{perm}}
\global\long\def\usat{\mathrm{usat}}
\global\long\def\PM{\mathcal{PM}}
\global\long\def\O{O}

\global\long\def\M{\mathcal{M}}
\global\long\def\PerfMatch{\mathrm{PerfMatch}}

\usepackage{authblk}

\date{}
\title{Block Interpolation: \\ A Framework for Tight Exponential-Time Counting Complexity}

\author{Radu Curticapean\thanks{Supported by ERC Starting Grant PARAMTIGHT (No. 280152). Part of this work was done at Saarland University and appeared in the author's PhD thesis \cite{DBLP:phd/dnb/Curticapean15}. Another part was done while the author was visiting the Simons Institute for the Theory of Computing.}}

\affil{Institute for Computer Science and Control, \\ Hungarian Academy of Sciences (MTA SZTAKI)}

\begin{document}

\maketitle

\begin{abstract}
We devise a framework for proving tight lower bounds under the counting exponential-time hypothesis $\sharpETH$ introduced by Dell~et~al.~(ACM Transactions on Algorithms, 2014). 
Our framework allows us to convert classical $\sharpP$-hardness results for counting problems into tight lower bounds under $\sharpETH$,
thus ruling out algorithms with running time $2^{o(n)}$ on graphs with $n$ vertices and $O(n)$ edges.
As exemplary applications of this framework, we obtain tight lower bounds under $\sharpETH$ for the evaluation of the zero-one permanent, the matching polynomial, and the Tutte polynomial on all non-easy points except for one line. This remaining line was settled very recently by Brand~et~al.~(IPEC 2016).
\end{abstract}

\section{Introduction}

Counting complexity is a classical subfield of complexity theory,
launched by Valiant's seminal paper \cite{V79} that introduced the
class $\sharpP$ and proved $\sharpP$-hardness of the zero-one permanent,
a problem equivalent to counting
perfect matchings in a bipartite graph. This initial breakthrough
spawned an ongoing research program that systematically studies the
complexity of computational counting problems, and many results in
this area can be organized as dichotomy results. Such results show
that, among problems that can be expressed in certain rich frameworks,
each problem is either polynomial-time solvable or $\sharpP$-hard.
Moreover, these results often give criteria for deciding which side
of the dichotomy a given problem occupies. For instance, a full dichotomy
was shown for the problems of counting solutions to constraint-satisfaction
problems \cite{DBLP:conf/icalp/Bulatov08,DBLP:conf/stoc/CaiC12},
and similar results are known for large subclasses of so-called Holant problems \cite{DBLP:conf/soda/CaiLX11,DBLP:journals/siamcomp/CaiGW16},
and for the evaluation of graph polynomials such as the Tutte polynomial
\cite{JVW90} and the cover polynomial \cite{BD07,DBLP:conf/mfcs/BlaserC11}.

Over the course of the counting complexity program, 
it became clear that most interesting counting problems are $\sharpP$-hard, 
and that the class of polynomial-time solvable problems is rather limited,
nevertheless containing some surprising examples, 
such as counting perfect matchings in planar graphs, 
counting spanning trees,
and problems amenable to holographic algorithms \cite{Valiant2008}. 
To attack the large body of hard problems, 
several relaxations were studied, 
such as approximate counting \cite{DBLP:journals/jacm/JerrumSV04,DBLP:journals/siamcomp/GoldbergJ14,GJ08Tutte},
counting modulo fixed numbers \cite{Hertrampf1990,Valiant2006},
and counting on restricted graph classes, such as planar and/or $3$-regular graphs \cite{DBLP:journals/siamcomp/Vadhan01,DBLP:conf/tamc/CaiLX09}.

In this paper, we follow an avenue of relaxations recently introduced
by Dell et al.~\cite{DBLP:journals/talg/DellHMTW14} and consider
the possibility of sub-exponential exact algorithms for counting problems.
More precisely, we rule out such algorithms for various counting problems under popular complexity-theoretic
assumptions. For instance, we can clearly count perfect matchings
on $m$-edge graphs in time $2^{O(m)}$ by brute-force, but is there
a chance of obtaining a running time of $2^{o(m)}$? An unconditional
negative answer would imply the separation of $\mathsf{FP}$ and
$\sharpP$, so our results need to rely upon additional hardness assumptions:
We build upon the \emph{exponential-time hypothesis}
$\sharpETH$, introduced in \cite{DBLP:journals/talg/DellHMTW14},
which we may consider for now as the hypothesis that the satisfying assignments to $3$-CNF formulas
$\varphi$ on $n$ variables cannot be counted
in time $2^{o(n)}$. This hypothesis is trivially implied by the better-known
and widely-believed decision version $\ETH$, introduced in \cite{IPZ01,Impagliazzo.Paturi2001},
which assumes the same lower bound for \emph{deciding }the satisfiability
of $\varphi$.

Dell et al.~\cite{DBLP:journals/talg/DellHMTW14} were able to prove almost-tight lower bounds under $\sharpETH$ for a variety of counting problems:
For instance, they could rule out algorithms with running time $2^{o(n/\log n)}$
for the zero-one permanent on graphs with $n$ vertices and $\O(n)$
edges. Similar lower bounds were shown for counting vertex covers,
and for most points of the Tutte polynomial. 

\subsection{Hardness via polynomial interpolation}

The lower bounds in \cite{DBLP:journals/talg/DellHMTW14} are obtained via polynomial interpolation, 
one of the most prominent techniques for non-parsimonious reductions between counting problems
\cite{L86,JVW90,DBLP:journals/siamcomp/Vadhan01,DBLP:conf/tamc/CaiLX09,DBLP:conf/iwpec/HusfeldtT10,DBLP:conf/iwpec/Hoffmann10,DBLP:journals/talg/DellHMTW14}.
To illustrate this technique, and for the purposes of further exposition,
let us reduce counting \emph{perfect matchings} to counting \emph{matchings}
(that are not necessarily perfect), using a standard argument similar
to \cite{DBLP:journals/siamcomp/Vadhan01}. In the following, let
$G$ be a graph with $n$ vertices. We wish to obtain the number of perfect matchings in $G$ by querying an oracle for counting matchings in arbitrary graphs.
\begin{description}
\item [{Step~1~\textendash ~Set up interpolation:}] For $k\in\mathbb{N}$,
let $m_{k}$ denote the number of matchings with exactly $k$ unmatched
vertices in $G$. In particular, $m_{0}$ is equal to the number of
perfect matchings in $G$. For an indeterminate $x$, define a polynomial $\mu$ via
\begin{equation}
\mu(x)=\sum_{k=0}^{n}m_{k}\cdot x^{k}\label{eq: match-poly-firstdef}
\end{equation}
and observe that its degree is $n$. Hence, we could use Lagrange interpolation
to recover all its coefficients if we were given the evaluations of $\mu$ at $n+1$ distinct input points.
In particular, this would give us the constant coefficient $m_{0}$, which counts the number of perfect matchings in $G$.

\item [{Step~2~\textendash ~Evaluate the polynomial with gadgets:}] We can evaluate $\mu(t)$
at points $t\in\mathbb{N}\setminus \{0\}$ by a reduction to counting
matchings: For $t\in\mathbb{N}$ with $t\geq1$, define a graph $G_{t}$
from $G$ by adding, for each vertex $v\in V(G)$, a \emph{gadget} 
that consists of an independent set of $t-1$ fresh vertices
together with edges from all of these vertices to $v$.
Then it can be checked that $\mu(t)$ is equal to the number of matchings in $G_{t}$: Each
matching in $G$ with exactly $k$ unmatched vertices can be extended
to $t^{k}$ matchings in $G_{t}$ by including up to one gadget edge
at each unmatched vertex.
\end{description}

In summary, by evaluating the polynomial $\mu(t)$ for all $t\in\{1,\ldots,n+1\}$ via
gadgets and an oracle for counting matchings in $G_{t}$, we can use
Lagrange interpolation to obtain $m_{0}$. 
This gives a polynomial-time Turing reduction from counting perfect matchings to counting matchings,
transferring the $\sharpP$-hardness of the former problem to the latter.

Furthermore, the above argument can also be used to derive a lower bound for counting
matchings, which is however far from being tight: If the running time
for counting perfect matchings on $n$-vertex graphs has a lower bound
of $2^{\Omega(n)}$, then only a $2^{\Omega(\sqrt{n})}$ lower bound
for counting matchings follows from the above argument, since the reduction incurs a quadratic blowup. 
This is because $G_{n+1}$ has a gadget of size $O(n)$ at each vertex, and thus $O(n^{2})$
vertices in total.

Following the same outline as above, but using more sophisticated
gadgets with $O(\log^{c}n)$ vertices, similar reductions for various problems were obtained in \cite{DBLP:conf/iwpec/Hoffmann10,DBLP:conf/iwpec/HusfeldtT10,DBLP:journals/talg/DellHMTW14},
implying $2^{\Omega(n/\log^{c}n)}$ lower bounds for these problems,
which are however still not tight. In particular, these reductions
share the somewhat unsatisfying commonality that they ``leak'' hardness: 
Tight lower bounds
for the source problems of computing specific hard coefficients in a polynomial became less
tight over the course of the reduction.

\subsection{The limits of interpolation}

Let us say that a reduction is \emph{gadget-interpolation-based} if it proceeds
along the two steps sketched above: First encode a hard problem into
the coefficients of a polynomial $p$, then find gadgets that can
be ``locally'' placed at vertices or edges so as to evaluate
$p(\xi)$ at sufficiently many points $\xi$.
Finally use Lagrange interpolation to recover $p$ from these
evaluations. As remarked before, this is a well-trodden route for
$\sharpP$-hardness proofs. However, when taking this
route to prove lower bounds under $\sharpETH$, we run into the following
obstacles:
\begin{enumerate}
\item Gadget-interpolation-based
reductions typically yield polynomials $p$ of degree $n=|V(G)|$, hence
require $n+1$ evaluations of $p$ at \emph{distinct} points, and
thus in turn require $n+1$ \emph{distinct} gadgets to be placed at
vertices of $G$. But since there are only finitely many simple graphs
on $O(1)$ vertices, the size of such gadgets must necessarily grow
as some unbounded function $\alpha(n)$. Thus, any gadget-interpolation-based
reduction can only yield $2^{\Omega(n/\alpha(n))}$ time lower bounds
for some unbounded function $\alpha\in\omega(1)$, but such bounds
are typically not tight.
\item Additionally, such reductions issue only polynomially
many queries to the target problem. This is required for the setting
of $\sharpP$-hardness, but it is nonessential in exponential-time
complexity: To obtain a lower bound of $2^{\Omega(n)}$, we might
as well use a reduction that requires $2^{o(n)}$ time and issues
$2^{o(n)}$ queries to the target problem, provided that the graphs
used in the oracle queries have only $\O(n)$ vertices. However, it
is a priori not clear how to exploit this additional freedom.
\end{enumerate}
These two issues are immanent to every known lower bound under $\sharpETH$
and have ruled out tight lower bounds of the form $2^{\Omega(n)}$
so far.

\subsection{The block-interpolation framework}

In this paper, we circumvent the barriers mentioned above by introducing
a framework that allows to apply the full power of sub-exponential
reductions to counting problems. 
To this end, we use a simple trick based on \emph{multivariate} polynomial
interpolation, which we dub \emph{block-interpolation}: In this setting,
we do not use a univariate polynomial $p$ of degree $n$ in the reduction, but rather
a \emph{multivariate} polynomial $\pmult p$ that we can easily obtain
from $p$. This polynomial $\pmult p$ also has total degree $n$,
but it has only maximum degree $\nicefrac{1}{\epsilon}$ in each individual
indeterminate, for any $\epsilon > 0$ we choose. By making sure that $\epsilon$ is small enough, we can interpolate
the polynomial $\pmult p$ from $2^{o(n)}$ evaluations.

While this increases the number of evaluations significantly,
we obtain the following crucial benefit: Each evaluation $\pmult p(\xi)$ required
for the interpolation can be performed at a point $\xi$ whose individual
entries are contained in a fixed set of size $\nicefrac{1}{\epsilon}+1$.
This will enable us to compute $\pmult p(\xi)$ by attaching only
$\nicefrac{1}{\epsilon}+1$ \emph{distinct} gadgets to $G$. The catch
here is that different vertices may obtain different gadgets, which
was not feasible in the univariate setting. 

This way, we overcome
the two above-mentioned limitations of gadget-interpolation-based
reductions while simultaneously staying as close as possible to the
outline of such reductions.
Consequently, we can phrase our technique as a general framework
that can be used to convert a large body of existing $\sharpP$-hardness
proofs into tight lower bounds under $\sharpETH$. Curiously enough,
the growth of the gadgets used in the original proofs is \emph{irrelevant}
to our framework, as only a constant number of gadgets will be used throughout the reduction.
This allows us to use luxuriously large gadgets and it shortcuts the
need for the involved and resourceful gadget constructions used, e.g.,
for simulating weights in the Tutte polynomial \cite{DBLP:journals/talg/DellHMTW14,DBLP:conf/iwpec/HusfeldtT10}
or in the independent set polynomial \cite{DBLP:conf/iwpec/Hoffmann10}.
More importantly, our bounds are tight.

\subsection{Applications of the framework}

To showcase our framework, we show that $\sharpETH$ rules out algorithms
with running time $2^{o(n)}$ for several classical problems on unweighted
simple graphs $G$ with $n$ vertices and $O(n)$ edges, which we
call sparse graphs. All of the considered problems admit trivial $2^{O(n)}$ time algorithms
on such graphs. It should be noted that it is crucial to obtain hardness
results for sparse graphs: Many reductions between counting
problems proceed by placing gadgets at \emph{edges}, and this would
map non-sparse graphs with $\omega(n)$ edges to target graphs on
$\omega(n)$ vertices, thus failing to provide $2^{\Omega(n)}$ time lower bounds
for the target problem.

More precisely, we show the following slightly stronger statements: Assuming
$\sharpETH$, each of the considered problems admits constants $\epsilon,C>0$
such that the problem cannot be solved in time
$O(2^{\epsilon n})$ on $n$-vertex graphs, even when these graphs have at
most $Cn$ edges. This clearly implies the claimed statements, and
we will elaborate on this in Section~\ref{subsec:Exponential-time-complexity}.
\begin{thm}
\label{thm: Main PM}Assuming $\sharpETH$, counting
perfect matchings admits no $2^{o(n)}$ time algorithm, even for graphs
that are bipartite, sparse, and unweighted.
\end{thm}
In \cite{DBLP:journals/talg/DellHMTW14}, only a lower bound of $2^{\Omega(n/\log n)}$ under $\sharpETH$
was shown for this problem. Tight lower bounds
of $2^{\Omega(n)}$ were obtained only (a) under the randomized version
$\rETH$ of $\ETH$, which implies $\ETH$ and thus in turn $\sharpETH$,
but no converse direction is known, or (b) under $\sharpETH$, but
by introducing negative edge weights. Such edge weights are generally
worrying, because it is a priori unclear how to remove them in reductions
to other problems.

By reduction from Theorem~\ref{thm: Main PM}, we then obtain a hardness
result for the matching polynomial, as defined in (\ref{eq: match-poly-firstdef}), and a similar graph polynomial, the independent set polynomial.
We will provide the precise definitions of the matching and independent set polynomials and
their associated evaluation problems in Section~\ref{subsec:Graph-polynomials}.

\begin{thm}
\label{thm: Main Match}Assuming $\sharpETH$, the problem of evaluating
the matching polynomial $\mu(G;\xi)$ admits no $2^{o(n)}$ time algorithm
at all fixed $\xi\in\mathbb{Q}$, even on graphs that are sparse and
unweighted. The same holds for the independent set polynomial $I(G;\xi)$
at all fixed $\xi\in\mathbb{Q}\setminus\{0\}$. 
\end{thm}
Both statements hold in particular at $\xi=1$, where these polynomials
simply count matchings, and independent sets, respectively. No lower
bounds for $\mu(G;\xi)$ are stated in the literature. In \cite{DBLP:conf/iwpec/Hoffmann10},
a lower bound of $2^{\Omega(n/\log^{3}n)}$ for $I(G;\xi)$ was shown
at general $\xi\in\mathbb{Q}\setminus\{0\}$, and a lower bound of $2^{\Omega(n)}$
was shown at $\xi=1$, but neither of these bounds apply to sparse graphs. 

Finally, we show lower bounds for the Tutte polynomial. Again, the
formal definition of this graph polynomial will be provided in Section~\ref{subsec:Graph-polynomials}.
\begin{thm}
\label{thm: Main Tutte}Assuming $\sharpETH$, the Tutte polynomial
$T(x,y)$ for fixed points $(x,y)\in\mathbb{Q}^{2}$ cannot be evaluated
in time $2^{o(n)}$ on sparse simple graphs if 

\begin{itemize}
\item $y\neq1$, and 
\item $(x,y)\notin\{(1,1),(-1,-1),(0,-1),(-1,0)\}$, and 
\item $(x-1)(y-1)\neq1$.
\end{itemize}

\end{thm}
In \cite{DBLP:journals/talg/DellHMTW14}, only lower bounds of the
type $2^{\Omega(n/\log^{c}n)}$ could be shown for the Tutte polynomial on sparse simple graphs.
Please consider \cite[Figure~1]{DBLP:journals/talg/DellHMTW14} for
a plot of the points covered by Theorem~\ref{thm: Main Tutte}: Our
lower bound applies at all points for which these authors show any
super-polynomial lower bound.
While our theorem leaves open the non-easy points on the line $y=1$, that is, all points $(x,1)$ with $x\neq 1$,
tight lower bounds for these points have been found by Brand~et~al.~\cite{DBLP:conf/iwpec/BrandDR16} since the conference version of the present paper was published.
We thus have:
\begin{thm}[\cite{DBLP:conf/iwpec/BrandDR16}]
\label{thm: Tutte full}
If $(x,y)\in\mathbb{Q}^{2}$ is such that $(x-1)(y-1)=1$ or $(x,y)\in\{(1,1),(-1,-1),(0,-1),(-1,0)\}$, 
then $T(x,y)$ can be computed in polynomial time.
Otherwise, there is no $2^{o(n)}$ time algorithm for computing $T(x,y)$ on sparse simple graphs unless $\sharpETH$ fails.
\end{thm}

\subsection*{Organization of this paper}

In Section~\ref{sec:Preliminaries}, we survey the necessary preliminaries
from the theory of graph polynomials, polynomial interpolation, and
exponential-time complexity. Then, in Section~\ref{sec:Framework},
we introduce the interpolation framework used for later results. In
Section~\ref{sec:Applications}, we prove Theorems \ref{thm: Main PM}-\ref{thm: Main Tutte}
as applications of this framework.

\section{\label{sec:Preliminaries}Preliminaries}

For $k\in\mathbb{N}$, we abbreviate $[k]=\{1,\ldots,k\}$.
For sets $X$, write ${X \choose 2}$ for the set
of all unordered pairs with elements from $X$. The graphs in this
paper are finite, undirected, and simple. If $G$ is a graph, we implicitly
assume that $V(G)=[n]$ for some $n\in\mathbb{N}$, and consequently $E(G)\subseteq{[n] \choose 2}$.
Graphs may feature edge- or vertex-weights within intermediate steps
of arguments, but all such weights will ultimately be removed to obtain
hardness results on \emph{unweighted} graphs. 

For simplicity, we phrase our results using only rational numbers,
but they could be easily adapted to the real or complex numbers,
provided that these numbers are represented properly. 
We also write $x\gets y$ for substituting the expression $y$ into an indeterminate $x$. 

\subsection{\label{subsec:Graph-polynomials}Graph polynomials}

Our arguments and statements of results use \emph{graph polynomials},
which are functions that map graphs $G$ to polynomials $p(G)\in\mathbb{Q}[\pmult x]$,
where $\pmult x$ is some set of indeterminates. They are usually
defined such that isomorphic graphs $G$ and $G'$ are required to
satisfy $p(G)=p(G')$, but we ignore this restriction for our purposes.
As a notational convention, we abbreviate $p(G;\xi)=(p(G))(\xi)$ for the
evaluation of the polynomial $p(G)$ at a point $\xi$.

The arguably most famous graph polynomial is the Tutte polynomial \cite{brylawski1982tutte},
which we define in the following, along with the matching polynomial
and the independent set polynomial \cite{IndependenceSurvey}.
\begin{defn}
[Matching and independent set polynomials]\label{def: massendefinition}
Let $G$ be a graph and let $\M[G]$ denote the set of (not necessarily
perfect) matchings in $G$, that is, edge-subsets that are vertex-disjoint. For $M\in\M[G]$, let $\usat(G,M)$ denote
the set of unmatched vertices of $G$ in $M$. Then we define the
\emph{matching polynomial} $\mu$ (also called \emph{matching defect
polynomial}) as 
\[
\mu(G;x)=\sum_{M\in\M[G]}x^{|\usat(G,M)|}.
\]

Similarly, let $\mathcal{I}[G]$ denote the independent sets of $G$.
Then the \emph{independent-set polynomial }$I$ is 
\[
I(G;x)=\sum_{S\in\mathcal{I}[G]}x^{|S|}.
\]
\end{defn}

Note that both the matching polynomial and the independent-set polynomial
are weighted sums over its eponymous structures. A similar definition
applies in the case of the Tutte polynomial, but the weights are more intricate.

\begin{defn}
[Tutte polynomial] For a subset $A\subseteq E(G)$, let $k(G,A)$
denote the number of connected components in the edge-induced subgraph
$G[A]$. Then define the classical parameterization of the Tutte
polynomial as 
\[
T(G;x,y)=\sum_{A\subseteq E(G)}(x-1)^{k(G,A)-k(G,E)}(y-1)^{k(G,A)+|A|-|V|}.
\]
We also use the \emph{random-cluster formulation} of the Tutte polynomial:
\[
Z(G;q,w)=\sum_{A\subseteq E(G)}q^{k(G,A)}w^{|A|}.
\]
\end{defn}
The polynomials $Z$ and $T$ are essentially different parameterizations
of each other: As noted in \cite{DBLP:journals/talg/DellHMTW14},
with $q=(x-1)(y-1)$ and $w=y-1$, we have
\begin{equation}
T(G;x,y)=(x-1)^{-k(G,E)}(y-1)^{-|V(G)|} \cdot Z(G;q,w).\label{eq: tutte-relations}
\end{equation}
We will mostly use the random-cluster formulation of the Tutte polynomial.

\begin{defn}
[Perfect matching polynomial and permanent] 
Recall that we assume that for each graph $G$, there is some $n \in \mathbb N$ such that $V(G)=[n]$ and $E(G)\subseteq{V(G) \choose 2}$.
For $e\in{\mathbb{N} \choose 2}$,
let $x_{e}$ be an indeterminate. We write $\PM[G]$ for the set of
perfect matchings of $G$. Then the perfect matching polynomial is
defined as 
\[
\PerfMatch(G)=\sum_{M\in\PM[G]}\prod_{e\in M}x_{e}.
\]
See also \cite[Section 3]{Valiant2008}, and note that only finitely
many indeterminates are present in $\PerfMatch(G)$. If $G$ is bipartite,
we also denote $\PerfMatch(G)$ by the \emph{permanent} $\perm(G)$.
In doing so, we slightly abuse notation, since the permanent is defined on matrices $A$,
but we implicitly consider the bi-adjacency matrix $A$ of the bipartite graph $G$ when speaking of the permanent of $G$.
\end{defn}
For any graph polynomial $p$, we define two computational problems
$\coeff p$ and $\eval p$, and a family of problems $\evalfix Sp$
for subsets $S\subseteq\mathbb{Q}$.
\begin{lyxlist}{00.00.0000}
\item [{$\coeff p:$}] On input $G$, output the list of all coefficients of $p(G)$.
In this paper, we will consider this problem only for univariate graph
polynomials.
\item [{$\eval p:$}] On input $G$ and an arbitrary point $\xi$, evaluate
$p(G;\xi)$. Here, $\xi$ is to be considered as a rational-valued assignment to
the indeterminates of $p(G)$. We will often consider $\xi$ as vertex- or edge-weights that are
substituted into the indeterminates of $p(G)$.
\item [{$\evalfix Sp:$}] On input $G$ and a point $\xi$ whose coordinate-wise
entries are all from $S$, evaluate $p(G;\xi)$. The problem differs
from $\eval p$ in that only specific points $\xi$ are allowed as
input. Note that, if $p$ is univariate and $S=\{a\}$ is a singleton
set, then $\evalfix Sp$ simply asks to compute $p(G;a)$ for fixed
$a$ on input $G$. We write $\evalfix ap$ in this case.
\end{lyxlist}
\begin{example}
For the matching polynomial $\mu$, the problem $\eval{\mu}$ asks
to evaluate $\mu(G;\xi)$ when given as input a graph $G$ and a number
$\xi$. For fixed $a$, the problem $\evalfix a{\mu}$ asks to evaluate
$\mu(G;a)$ on input $G$. For instance, the problem $\evalfix 0{\mu}$
asks to count perfect matchings in a graph.
\end{example}
Rather than evaluating a multivariate graph polynomial $\pmult p$
like $\PerfMatch$ on an unweighted graph $G$ and a point $\xi$,
we often annotate edges/vertices of $G$ with the entries of $\xi$,
assuming $V(G)$ and $E(G)$ to be ordered. We then speak of evaluating
$\pmult p(G')$ on the weighted graph $G'$ derived from $G$ and
$\xi$ this way.

\subsection{Multivariate polynomial interpolation}

Given a univariate polynomial $p\in\mathbb{Q}[x]$ of degree $n$,
we can use Lagrange interpolation to compute the coefficients of $p$
when provided with the set of evaluations $\{(\xi,p(\xi))\mid\xi\in\Xi\}$
for any set $\Xi\subseteq\mathbb{Q}$ of size $n+1$. It is known that polynomial interpolation can be
generalized to multivariate polynomials $p\in\mathbb{Q}[\mathbf{x}]$,
for instance, if $\Xi$ is a sufficiently large grid.
\begin{lem}
\label{lem: multivar-interpolation}Let $p\in\mathbb{Q}[x_{1},\ldots,x_{n}]$
be a multivariate polynomial such that, for all $i\in[n]$, the degree
of $x_{i}$ in $p$ is bounded by $d_{i}\in\mathbb{N}$. 
Furthermore, assume we are given sets $\Xi_{i}\subseteq\mathbb{Q}$ for $i\in[n]$ such that $|\Xi_{i}|=d_{i}+1$ for all $i\in[n]$. Consider the cartesian product of these sets, that is,
\[
\Xi:=\Xi_{1}\times\ldots\times\Xi_{n}.
\]
Then we can compute the coefficients of $p$ with $O(|\Xi|^{3})$
arithmetic operations when given as input the set 
\[
\{(\xi,p(\xi))\mid\xi\in\Xi\}.
\]
\end{lem}
\begin{proof}
For $s,t,s',t'\in\mathbb{N}$ and matrices $A\in\mathbb{Q}^{s\times t}$
and $B\in\mathbb{Q}^{s'\times t'}$, we write $A\otimes B$ for the
Kronecker product of $A$ and $B$, which is the matrix $A\otimes B\in\mathbb{Q}^{s\cdot s'\times t\cdot t'}$
whose rows are indexed by $[s]\times[s']$, whose columns are indexed
by $[t]\times[t']$, and which satisfies 
\[
(A\otimes B)_{(i,i'),(j,j')}=A_{i,j}\cdot B_{i',j'}\quad\mbox{for }(i,i')\in[s]\times[s']\mbox{ and }(j,j')\in[t]\times[t'].
\]

For $\ell\in[n]$, enumerate $\Xi_{\ell}=\{a_1^{(\ell)},\ldots,a_{d_\ell+1}^{(\ell)}\}$
and let $A^{(\ell)}$ denote the
Vandermonde matrix of dimensions  $(d_{\ell}+1)\times(d_{\ell}+1)$ with $A_{i,j}^{(\ell)}=(a_{i}^{(\ell)})^{j}$ for all
$i,j\in[d_{\ell}+1]$. It is well-known that each Vandermonde matrix
$A^{(\ell)}$ for $\ell\in[n]$ has full rank, provided that $a_{i}^{(\ell)}\neq a_{i'}^{(\ell)}$
for all $i\neq i'$. This condition is guaranteed in our setting.
Now define 
\[
A:=A^{(1)}\otimes\ldots\otimes A^{(n)}.
\]
Since each matrix $A^{(\ell)}$ for $\ell\in[n]$ has full rank, so does
the matrix $A$, by an elementary property of the rank of Kronecker
products \cite[Corollary 13.11]{Laub2004}. 

Let $\mathbf{c}$ denote
the vector that lists the coefficients of $p$ in lexicographic order\footnote{This vector includes the coefficients of all monomials with degree
at most $d_{i}$ in $x_{i}$, even if some of these coefficients may
be zero.}, and let $\mathbf{v}$ denote the vector that lists the evaluations
$p(\xi)$ for $\xi\in\Xi$ in lexicographic order. Then it can be
verified that 
$A\mathbf{c}=\mathbf{v}$.
Since $A$ has full rank, this system of linear equations can be solved
with $O(|\Xi|^{3})$ arithmetic operations for $\mathbf{c}$, and
we obtain the coefficients of $p$.
\end{proof}
\begin{rem}
By exploiting faster methods for solving linear systems of equations,
the running time above could be lowered from $O(|\Xi|^{3})$ to $O(|\Xi|^{\omega})$
operations, where $\omega<2.4$ is the exponent of matrix multiplication.
This is however non-essential for our reductions.
\end{rem}

\subsection{Exponential-time complexity}
\label{subsec:Exponential-time-complexity}
We build upon the counting exponential-time hypothesis $\sharpETH$
introduced in \cite{DBLP:journals/talg/DellHMTW14}, which is a variant of
the corresponding hypothesis $\ETH$ for decision problems \cite{Impagliazzo.Paturi2001,IPZ01}. 
\begin{defn}
\label{def: ETH}The \emph{counting exponential-time hypothesis} $\sharpETH$
is the following claim: There is a constant $\epsilon>0$ such that
no deterministic algorithm with running time $O(2^{\epsilon n})$
can count the satisfying solutions of $3$-CNF formulas $\varphi$
with $n$ variables.
\end{defn}
Note that $\sharpETH$ rules out $2^{o(n)}$ time algorithms for counting
satisfying assignments of $3$-CNF formulas with $n$ variables. In
fact, $\sharpETH$ is often stated as claiming precisely this lower
bound. However, this latter statement is a priori not equivalent to
$\sharpETH$, as there could be, say, an uncomputable sequence of
$O(2^{\epsilon n})$ time algorithms with $\epsilon\to0$ for counting
satisfying assignments. For this reason, some authors choose to characterize
the original definitions of $\ETH$ and $\sharpETH$ as \emph{nonuniform}
\cite{DBLP:conf/iwpec/ChenEF12}.

A particularly useful tool for proving lower bounds under $\sharpETH$
is the \emph{sparsification lemma}, which was first shown
for the decision version $\ETH$ \cite[Corollary 1]{IPZ01} and later
adapted to counting problems \cite[Lemma A.1]{DBLP:journals/talg/DellHMTW14}.
This result allows us to assume that the formulas $\varphi$ in Definition~\ref{def: ETH}
are sparse, i.e., even an $2^{o(m)}$ time algorithm would refute
$\sharpETH$, where $m$ is the number of clauses of $\varphi$. Note
that this indeed strengthens $\sharpETH$, as we may assume 
$n\leq m$, whereas we can a priori only guarantee $m=O(n^{3})$ for $3$-CNF formulas.
\begin{thm}
Assuming $\sharpETH$, there is a constant $\epsilon>0$ such that
no deterministic algorithm with running time $O(2^{\epsilon m})$ can count the
satisfying solutions of $3$-CNF formulas $\varphi$ with $m$ clauses.
\end{thm}
This theorem is shown by an application of so-called sub-exponential
reduction families \cite[Section 1.1.4]{IPZ01}. In the following
definition, we adapt these reductions for our particular applications
involving graph polynomials. That is, we restrict our definition to
problems $\A$ and $\B$ that receive graphs as inputs, and we ensure
that the instances generated by the reduction are sparse.
\begin{defn}
\label{def: subexponential-reductions}A \emph{sub-exponential reduction
family} from problem $\A$ to $\B$ is an algorithm $\mathbb{T}$
with oracle access for $\B$. Its inputs are pairs $(G,\epsilon)$
where $G$ is an input graph for $\A$, and $\epsilon$ with $0<\epsilon\leq1$
is a running time parameter, such that

\begin{enumerate}
\item $\mathbb{T}$ computes $\A(G)$, and it does so in time $f(\epsilon)\cdot2^{\epsilon\cdot|V(G)|}$,
and 
\item $\mathbb{T}$ only invokes the oracle for $\B$ on graphs $G'$ with
at most $g(\epsilon)\cdot(|V(G)|+|E(G)|)$ vertices and edges.
\end{enumerate}
In these statements, $f$ and $g$ are computable functions that depend
only on $\epsilon$. We write $\A\,\leqSERF\,\B$ if such a reduction
exists.
\end{defn}
That is, the running time of $\mathbb{T}$ (and hence, the number
of oracle queries) can be chosen as $O(2^{\epsilon n})$ for arbitrarily
small $\epsilon$, but this comes at the cost of incurring a ``blowup
factor'' of $g(\epsilon)$ in the reduction images. It can be verified
that sub-exponential reductions preserve lower bounds,
see \cite[Section 1.1.4]{IPZ01}:
\begin{lem}
\label{lem: subexponential reductions work}Let $\A$ and $\B$ be
problems that satisfy $\A\leqSERF\B$, and assume that for all $\epsilon,C>0$,
there is an $O(2^{\epsilon n})$ time algorithm for $\B$ on graphs
with $n$ vertices and at most $Cn$ edges. Then, for all $\delta,D>0$,
there is an $O(2^{\delta n})$ time algorithm for $\A$ on graphs
with $n$ vertices and at most $Dn$ edges.
\end{lem}
\begin{cor}
\label{cor: subexponential reductions work, contraposition}If $\A\leqSERF\B$
and there are $\epsilon,C>0$ such that $\A$ cannot be solved in
time $O(2^{\epsilon n})$ on graphs with $n$ vertices and at most
$Cn$ edges, then there are $\delta,D>0$ such that $\B$ cannot be
solved in time $O(2^{\delta n})$ on graphs with $n$ vertices and
at most $Dn$ edges. 
\end{cor}
If the prerequisites of the above corollary hold, then it is in particular
true that there is a constant $D$ such that $\B$ cannot be solved in time $2^{o(n)}$ on graphs with
$n$ vertices and at most $Dn$ edges.

\section{\label{sec:Framework}The Block Interpolation Framework}

In this section, we show how to obtain tight lower bounds for evaluating graph polynomials under $\sharpETH$ by means
of multivariate polynomial interpolation. More specifically, for a
general class of univariate graph polynomials $p$, we show that,
for certain fixed $\xi\in\mathbb{Q}$, we can reduce the coefficient problem of $p$ to the evaluation problem of $p$ on $\xi$.
\begin{equation}
\coeff p\leqSERF\evalfix{\xi}p.\label{eq: ETH full-reduction}
\end{equation}

This is useful due to the following reasons: Firstly, many counting
problems can be expressed as evaluation problems $\evalfix{\xi}p$
for adequate graph polynomials $p$ and points $\xi$. For instance,
the Tutte polynomial collects an abundance of such examples.
Secondly, as discussed in the introduction, many classical $\sharpP$-hardness
proofs for $\evalfix{\xi}p$ first establish $\sharpP$-hardness for
$\coeff p$ and then reduce this to the evaluation problem by some
form of interpolation. In many cases, the classical $\sharpP$-hardness
proof for $\coeff p$ already yields a tight lower bound under $\sharpETH$.
Our technique then allows to transfer this lower bound to
$\evalfix{\xi}p$.

The remainder of this section is structured as follows: In Section~\ref{subsec:Admissible-graph-polynomials},
we first describe the ``format'' required from a univariate graph
polynomial $p$ for our framework to apply. Then we show in Section~\ref{subsec:Using-multivariate-interpolation}
how to perform the reduction
\begin{equation}
\coeff p\leqSERF\evalfix S{\pmult p},\label{eq: ETH coeff to eval_wt}
\end{equation}
where $\pmult p$ is a certain ``multivariate version'' of $p$, as mentioned in the introduction,
and $S\subseteq\mathbb{Q}$ is a set whose size depends only upon the running
time parameter $\epsilon$ in the sub-exponential reduction family,
but not on the size of the input graph. In Section~\ref{subsec:Weight-simulation},
we then show how to reduce 
\begin{equation}
\evalfix S{\pmult p}\leqSERF\evalfix{\xi}p\label{eq: ETH eval_wt to eval}
\end{equation}
by means of gadget families, provided that these families exist. Pipelining the
reductions (\ref{eq: ETH coeff to eval_wt}) and (\ref{eq: ETH eval_wt to eval})
then gives the full reduction (\ref{eq: ETH full-reduction}).

\subsection{\label{subsec:Admissible-graph-polynomials}Admissible Graph Polynomials}

Our framework applies to univariate graph polynomials $p$ that admit
a canonical multivariate generalization. More specifically, we call
$p$ \emph{subset-admissible} if $p$ is induced by a \emph{sieving
function} $\chi$ which filters the structures counted by $p$, and
a \emph{weight selector} $\omega$ which assigns a particular kind of 
weight to each of these structures. While this definition may seem abstract at first, we will soon observe that various popular graph polynomials can be expressed naturally in this form.
\begin{defn}
\label{def:admissible}Let $\mathcal{G}$ denote the set of all unweighted
simple graphs and recall that, for each graph $G\in\mathcal{G}$, we assume that there exists some $n \in \mathbb N$ such that $V(G)=[n]$
 and $E(G)\subseteq{[n] \choose 2}$. Let
$\mathcal{V}=\mathbb{N}$ denote the set of all possible vertices of unweighted
simple graphs, and let $\mathcal{E}={\mathbb{N} \choose 2}$ denote
the set of all possible edges of such graphs. We also write $\mathcal{F}=\mathcal{V}\cup\mathcal{E}$.

For a \emph{sieve function} $\chi:\mathcal{G}\times2^{\mathcal{F}}\to\mathbb{Q}$
and a \emph{weight selector} $\omega:\mathcal{G}\times2^{\mathcal{F}}\to2^{\mathcal{F}}$,
we define the graph polynomial \emph{induced} by the pair $(\chi,\omega)$ as
\begin{equation}
p_{\chi,\omega}(G;x)=\sum_{A\subseteq V(G)\cup E(G)}\chi(G,A)\cdot x^{|\omega(G,A)|}.\label{eq:general-poly}
\end{equation}
We say that $p$ is \emph{subset-admissible} if $p=p_{\chi,\omega}$ for
some $(\chi,\omega)$ as above.
\end{defn}
\begin{rem*}
Note that $\chi$ and $\omega$ may be partial functions in the above definition, since, e.g.,
the value of $\chi(G,A)$ is irrelevant if $A\not\subseteq V(G)\cup E(G)$.
\end{rem*}
In the following, we observe that the matching polynomial $\mu$ and the independent set polynomial $I$
from Definition~\ref{def: massendefinition} are subset-admissible.
It would be nice to show the same for the Tutte polynomials $T$ and
$Z$, but this fails for syntactic reasons, since we defined admissible polynomials to be univariate. Instead, we will work with restrictions of
$Z$ to $Z_{q}:=Z(q,\cdot)$ for fixed $q\in\mathbb{Q}$. 
\begin{example}
\label{exa: tutte-reformulated}Given a sentence $\phi$, define $[\phi]=1$
if $\phi$ is true, and $[\phi]=0$ otherwise. With this notion, the
matching polynomial $\mu$ is induced by
\begin{eqnarray*}
 & \chi: & (G,A)\mapsto[A\in\mathcal{M}[G]],\\
 & \omega: & (G,A)\mapsto\usat(G,A),
\end{eqnarray*}
and $I$ is induced similarly by $\chi:(G,A)\mapsto[A\in\mathcal{I}[G]]$ and
$\omega:(G,A)\mapsto A$.

For $q\in\mathbb{Q}\setminus\{0\}$, the polynomial $Z_{q}=Z(q,\cdot)$
is induced by $\chi:(G,A)\mapsto q^{k(G,A)}$ and $\omega:(G,A)\mapsto A$.
We stress again that $Z_{q}\in\mathbb{Q}[x]$ is a univariate restriction
of $Z$ for fixed $q\in\mathbb{Q}$.
\end{example}
Every subset-admissible graph polynomial of the form $p_{\chi,\omega}$ admits a canonical
\emph{multivariate generalization} $\pmult p{}_{\chi,\omega}$ on
indeterminates $\mathbf{x}=\{x_{a}\mid a\in\mathcal{F}\}$, which
is given by
\begin{equation}
\pmult p_{\chi,\omega}(G;\mathbf{x})=\sum_{A\subseteq V(G)\cup E(G)}\chi(G,A)\prod_{a\in\omega(G,A)}x_{a}.\label{eq: multivar-p}
\end{equation}
Please note that only finitely many indeterminates from $\mathbf{x}$
are present in $\pmult p_{\chi,\omega}(G)$ for any (finite) graph
$G$. Compare (\ref{eq: multivar-p}) to (\ref{eq:general-poly}):
It is clear that, when substituting
$x_{a}\gets x$ for all $a\in\mathcal{F}$, the multivariate polynomial $\pmult p_{\chi,\omega}$
coincides with the univariate polynomial $p_{\chi,\omega}$. Note also that $\pmult p$
is multilinear by definition. Similar multivariate generalizations
were known, e.g., for the Tutte polynomial \cite{Sokal2005}.
\begin{example}
\label{exa: strange-p-becomes-perfmatch}Consider the polynomial $p=p_{\chi,\omega}$
induced by 
\begin{eqnarray*}
 & \chi: & (G,A)\mapsto[A\in\PM[G]],\\
 & \omega: & (G,A)\mapsto\ A,
\end{eqnarray*}
This polynomial $p$ admits the simple expression $p(G)=m_{G}\cdot x^{|V(G)|/2}$,
where $m_{G}$ denotes the number of perfect matchings in $G$. Note that $p(G)$ contains at most one monomial. Its multivariate generalization however gives us the richer structure $\pmult p(G)=\PerfMatch(G)$.

Furthermore, for fixed $q\in \mathbb Q$, the multivariate generalization of $Z(q,\cdot)$ yields
the so-called \emph{multivariate Tutte polynomial }considered in \cite{Sokal2005}:\emph{
}
\[
\pmult Z(G;q,\cdot)=\sum_{A\subseteq E(G)}q^{k(G,A)}\prod_{e\in A}x_{e}.
\]
\end{example}
If $p$ is a univariate subset-admissible polynomial and $\pmult p$
is its multivariate generalization, then the following simple relation
holds between the coefficients of $p$ and $\pmult p$:
\begin{lem}
\label{prop: coefficients}For any monomial $\theta$, let $c_{\theta}$
denote the coefficient of $\theta$ in $\pmult p$. For $k\in\mathbb{N}$,
let $C_{k}$ denote the set of monomials with total
power $k$ in $\pmult p$. Then for any $k\in\mathbb{N}$, the coefficient of $x^{k}$ in
$p$ is equal to $\sum_{\theta\in C_{k}}c_{\theta}$.
\end{lem}
\begin{proof}
When substituting $x_{a}\gets x$ for all $a\in\mathcal{F}$, we obtain
$p$ from $\pmult p$, and the monomials transformed to $x^{k}$ are
precisely the members of $C_{k}$. This proves the claim by collecting the coefficients of these monomials.
\end{proof}

\subsection{\label{subsec:Using-multivariate-interpolation}First Reduction Step:
Multivariate Interpolation}

Let $p=p_{\chi,\omega}$ be a subset-admissible polynomial with multivariate generalization
$\pmult p$. For ease of presentation, we assume for now that $\omega:\mathcal{G}\times2^{\mathcal{F}}\to2^{\mathcal{E}}$,
that is, $\omega$ maps only into edge-subsets rather than subsets
of edges and vertices. The general case is shown identically, with
more notational overhead.

We perform the reduction $\coeff p\leqSERF\eval{\pmult p}$ by means of interpolation.
Recall that, in the univariate case, to obtain $p(G)$ for an $m$-edge
graph $G$, we require the evaluations of $p(G;\xi)$ at $m+1$ distinct
points $\xi$. For the multivariate generalization $\pmult p$, we
could in principle interpolate via Lemma~\ref{lem: multivar-interpolation}: Since
$\pmult p$ is multilinear, this requires the evaluations of $\pmult p(G;\xi)$
on a grid $\Xi$ with two distinct values per coordinate, say $\Xi=[2]^{m}$.
By Lemma~\ref{prop: coefficients}, the coefficients of $p$ can
be obtained from those of $\pmult p$, so we could interpolate $\pmult p$
to recover $p$. 

While this detour seems extremely wasteful due to its $2^{m}$ (rather
than $m+1$) incurred evaluations, it yields the following reward:
For each variable $x_{e}$ in $\pmult p$, Lemma~\ref{lem: multivar-interpolation}
only requires us to substitute \emph{two} distinct values (or \emph{weights})
into $x_{e}$, whereas interpolation on $p$ requires $m+1$ distinct
substitutions to its only variable $x$. The small number of distinct
weights will prove very useful, since each such weight will be simulated
by a certain gadget. If there are only two weights to simulate, then
we require only two fixed gadgets, whose sizes are trivially bounded
by $O(1)$.

However, to interpolate $\pmult p$, we still need the prohibitively
large number of $2^{m}$ evaluations. To overcome this, we trade off
the number of evaluations with the numbers of distinct values per
edge, and thus, with the size of the gadgets ultimately required.
To this end, we group the edges into blocks and treat all edges within each block identically,
similar to \cite{DBLP:journals/talg/CyganDLMNOPSW16}.
At this point, the full power of sub-exponential reduction families from
Definition~\ref{def: subexponential-reductions} is used crucially.
\begin{lem}
\label{lem: Weighted block-wise interpolation}Let $p$ be subset-admissible,
with multivariate generalization $\pmult p$, and let $W=(w_{0},w_{1},\ldots)$
be an infinite recursively enumerable sequence of pairwise distinct
numbers in $\mathbb{Q}$.

Then $\coeff p\leqSERF\evalfix W{\pmult p}$ holds by a reduction
that satisfies the following for all inputs $(G,\epsilon)$: There
is some number $d=d(\epsilon)$ depending only upon $\epsilon$, such that all oracle queries
$\pmult p(G')$ are asked only on graphs $G'$ obtained from $G$
by introducing edge-weights from $W_{d}=\{w_{0},\ldots,w_{d}\}$.
\end{lem}
When invoking Lemma~\ref{lem: Weighted block-wise interpolation},
the list $W$ contains the weights that can be simulated by gadgets.
Note that \emph{any} such list $W$ can be used if $W$ is infinite
and recursively enumerable. Furthermore, note that $\pmult p$ is
evaluated only on edge-weighted versions of $G$ itself; properties
such as bipartiteness of $G$ or its size are hence trivially preserved.
\begin{proof}
[Proof of Lemma \ref{lem: Weighted block-wise interpolation}]Let
$d\in\mathbb{N}$ be a parameter, to be chosen later depending on
$\epsilon$, and let $G=(V,E)$ be an $m$-edge graph for which we
want to determine the coefficients of $p=p(G)$. Let 
\[
\pmult x=\{x_{e}\mid e\in E\}
\]
denote the indeterminates of $\pmult p$ and note that both $p$ and
$\pmult p$ have maximum degree $m$ by definition.

In the first step, partition $E$ into $t:=\lceil m/d\rceil$ sets
$E_{1},\ldots,E_{t}$ of size at most $d$ each (which we call \emph{blocks}),
using an arbitrary equitable assignment of edges to blocks. Define
new indeterminates 
\[
\pmult y=\{y_{1},\ldots,y_{t}\}
\]
and a new multivariate polynomial $\pmult q\in\mathbb{Q}[\pmult y]$
from $\pmult p$ by substituting 
\[
x_{e}\gets y_{i}\mbox{\quad}\mbox{for all }i\in[t]\mbox{ and }e\in E_{i}.
\]

We are working with three polynomials, namely $p$, $\pmult p$ and
$\pmult q$, summarized in Table~\ref{tab:interp-polynomials} for
convenience.
\begin{table}
\begin{centering}
\begin{tabular}{c|c|c|c|}
 & $p\in\mathbb{Q}[x]$ & $\pmult p\in\mathbb{Q}[\pmult x]$ & $\pmult q\in\mathbb{Q}[\pmult y]$\tabularnewline
\hline 
number of indeterminates & $1$ & $m$ & $t=\lceil m/d\rceil$\tabularnewline
\hline 
max.~degree per indeterminate & $m$ & $1$ & $d$\tabularnewline
\hline 
max.~number of monomials & $m+1$ & $2^{m}$ & $(d+1)^{t}$\tabularnewline
\hline 
\end{tabular}
\par\end{centering}
\caption{\label{tab:interp-polynomials}The polynomials $p$, $\protect\pmult p$
and $\protect\pmult q$ appearing in the proof of Lemma~\ref{lem: Weighted block-wise interpolation}.}
\end{table}
 While the total degree of $\pmult q$ is bounded by $m$, the degree
of each indeterminate $y_{i}$ in $\pmult q$ is bounded by $d$,
since each block contains at most $d$ edges. Hence, the number of
monomials in $\pmult q$ is at most $(d+1)^{t}=2^{d'm}$ with $d'=O(\log(d)/d)$.
Note that $d'\to0$ as $d\to\infty$.

We will ultimately obtain the coefficients of $\pmult q$ via interpolation,
but first, let us observe that the coefficients of $\pmult q$ allow us
to determine those of the univariate version $p$. Write $c_{k}^{p}$
for the coefficient of $x^{k}$ in $p$ and $c_{\theta}^{\pmult q}$
for the coefficient of the monomial $\theta$ in $\pmult q$. Analogously
to Lemma~\ref{prop: coefficients}, we have $c_{k}^{p}=\sum_{\theta\in C_{k}}c_{\theta}^{\pmult q}$
where $C_{k}$ for $k\in\mathbb{N}$ is the set of all monomials with
total power $k$ in $\pmult q$. This allows us to compute the coefficients
of $p$ from those of $\pmult q$.

It remains to describe how to obtain the coefficients of $\pmult q$.
For this, recall the definition of $W_{d}$ from the statement of the lemma. We
evaluate $\pmult q$ on the grid $\Xi=(W_{d})^{t}$ using the oracle
for $\evalfix W{\pmult p}$: For each $\xi\in\Xi$, substitute $y_{i}\gets\xi_{i}$
for all $i\in[t]$ to obtain an edge-weighted graph $G_{\xi}$ that
contains only weights from $W_{d}$, and for which we can thus compute
$\pmult p(G_{\xi})$ by an oracle call to $\evalfix W{\pmult p}$. 

Using $|\Xi|=(d+1)^{t}=2^{d'm}$ oracle calls and grid interpolation
via Lemma~\ref{lem: multivar-interpolation}, we obtain all coefficients
of $\pmult q$ with $O(2^{3d'm})$ arithmetic operations. By definition
of $\pmult p$ and $\pmult q$ and the set $W_d$, each arithmetic operation involves
numbers on at most $\O(m)\cdot g(d)$ bits for a computable function $g$.
Since $d'\to0$ as $d\to\infty$, we can pick $d$ large enough such
that $3d'\leq\epsilon$, and thus achieve running time $O(2^{\epsilon m})$.
No vertices or edges were added to $G$ during this reduction.
\end{proof}

\subsection{\label{subsec:Weight-simulation}Second Reduction Step: Weight Simulation
by Gadgets}

Lemma~\ref{lem: Weighted block-wise interpolation} gives
a sub-exponential reduction family from $\coeff p$ to $\eval{\pmult p}$
that maps instances $(G,\epsilon)$ to edge-weighted versions
obtained from $G$ by introducing $f(\epsilon)$ distinct edge-weights for some computable function $f$.
For the full reduction, this latter problem must be reduced to $\evalfix{\xi}p$
for \emph{fixed} $\xi\in\mathbb{Q}$. 

This may not work for all $\xi\in\mathbb{Q}$: For instance, the evaluation problem $\evalfix 0I$
for the independent-set polynomial $I$ at the point $0$ is trivial. We must
hence impose several conditions on $\xi$ to enable this reduction.
\begin{defn}
\label{def: Weight simulation}Let $p$ be subset-admissible, let
$\xi\in\mathbb{Q}$ and 

\begin{itemize}
\item let $W=(w_{0},w_{1},\ldots)$ be a sequence of pairwise distinct values
in $\mathbb{Q}$,
\item let $\mathcal{H}=(H_{0},H_{1},\ldots)$ be a sequence of \emph{edge-gadgets},
which are triples $(H,u,v)$ with a graph $H$ and \emph{attachment
vertices} $u,v\in V(H)$, and 
\item let $F:\mathcal{G}\times\mathbb{Q}\to\mathbb{Q}\setminus\{0\}$ be
a polynomial-time computable function, which we call a \emph{factor
function}.
\end{itemize}
If $G$ is an edge-weighted graph with weights from $W$, let $T(G)$ be the unweighted graph obtained
by replacing, for $i\in\mathbb{N}$, each edge $uv\in E(G)$ of weight
$w_{i}$ with a fresh copy of the edge-gadget $H_{i}$ by identifying $u,v$ across
$G$ and $H_{i}$. 

We say that \emph{$(\mathcal{H},F)$ allows to reduce $\evalfix W{\pmult p}$
to $\evalfix{\xi}p$ }if the following holds: Whenever $G$ is a graph
with edge-weights from $W$, then 
\begin{equation}
\pmult p(G)=\frac{p(T(G);\xi)}{F(G,\xi)}.\label{eq: Simulate weights}
\end{equation}
\end{defn}
\begin{rem}
\label{rem:vertex-weights}The same definition applies to vertex-weighted
graphs; here we use \emph{vertex-gadgets}, which are pairs $(H,v)$
with an attachment vertex $v\in V(H)$. Vertex-gadgets are inserted
at a vertex $v\in V(G)$ by identifying $v$ in $H$ and $G$.
\end{rem}
At the end of this subsection, we discuss several possible extensions
of this definition. As a first example for the notions introduced
in Definition~\ref{def: Weight simulation}, we consider (well-known)
edge-gadgets for $\PerfMatch$.
\begin{example}
\label{exa: gadgets-perfmatch}Let $p$ denote the polynomial from
Example~\ref{exa: strange-p-becomes-perfmatch}, whose multivariate generalization satisfies $\pmult p=\PerfMatch$.
For $k\in\mathbb{N}$, let $H_{k}$ be the edge-gadget obtained by placing
$k$ parallel edges between two fresh vertices $u$ and $v$ and subdividing each
edge twice. Let $\mathcal{H}=(H_{1},H_{2},\ldots)$, let $\mathbb{N}=(1,2,3,\ldots)$
and let $F$ denote the function mapping all inputs to $1$. Then
it can be seen that $(\mathcal{H},F)$ allows to reduce $\evalfix{\mathbb{N}}{\pmult p}$
to $\evalfix 1p$.
\end{example}
We easily obtain the following lemma from Definition~\ref{def: Weight simulation}.
\begin{lem}
\label{lem: weight-reduction}Let $W=(w_{0},w_{1},\ldots)$ and let
$(\mathcal{H},F)$ allow to reduce $\evalfix W{\pmult p}$ to $\evalfix{\xi}p$.
Let $G$ be a graph with edge-weights from $W$. Then we can 
compute $\pmult p(G)$ from $p(T(G);\xi)$ in polynomial time via (\ref{eq: Simulate weights}).
If $G$ has $n$ vertices and $m$ edges, and only contains
edge-weights $w_{i}$ with $i\leq t$ for some $t\in\mathbb{N}$,
then $T(G)$ has $O(n+sm)$ vertices and edges, where $s=\max_{i\in[t]}|V(H_{i})|+|E(H_{i})|$
depends only on $\mathcal{H}$ and $t$.
\end{lem}
By combining Lemmas~\ref{lem: Weighted block-wise interpolation}
and \ref{lem: weight-reduction}, we obtain the wanted reduction from
$\coeff p$ to $\evalfix{\xi}p$ at fixed points $\xi\in\mathbb{Q}$
and finish the set-up of our framework.
\begin{thm}
\label{thm: Block interpolation} Let $p$ be subset-admissible and
let $\xi\in\mathbb{Q}$ be fixed. Assuming $\sharpETH$, there are
constants $\epsilon,C>0$ such that the problem $\evalfix{\xi}p$
admits no $O(2^{\epsilon n})$ time algorithm on unweighted graphs
with $n$ vertices and at most $Cn$ edges, provided that the following
two conditions hold:

\begin{description}
\item [{Coefficient~hardness:}] Assuming $\sharpETH$, there are constants
$\delta,D>0$ such that $\coeff p$ admits no $O(2^{\delta n})$ time
algorithm on unweighted graphs with $n$ vertices and at most $Dn$
edges.
\item [{Weight~simulation:}] There is a recursively enumerable sequence
of pairwise distinct weights $W=(w_{0},w_{1},\ldots)$, a sequence
of gadgets $\mathcal{H}=(H_{0},H_{1},\ldots)$, and a function
$F$ such that $(\mathcal{H},F)$ allows to reduce $\evalfix W{\pmult p}$
to $\evalfix{\xi}p$.
\end{description}
\end{thm}
\begin{proof}
We present a sub-exponential reduction family from $\coeff p$ to
$\evalfix{\xi}p$. Given $\gamma>0$ and a graph $G$ with $n$ vertices
and $Dn$ edges, first apply Lemma~\ref{lem: Weighted block-wise interpolation}:
This way, we reduce $\coeff p$ to $O(2^{\gamma n})$ instances of
$\evalfix W{\pmult p}$ on graphs $G'$ that only use weights $w_{0},\ldots,w_{s}$
with $s=s(\gamma)$.

Since $(\mathcal{H},F)$ allows to reduce $\evalfix W{\pmult p}$
to $\evalfix{\xi}p$, we can invoke Lemma~\ref{lem: weight-reduction}
and reduce each instance $G'$ for $\evalfix W{\pmult p}$ to an instance
of $\evalfix{\xi}p$ on the graph $T(G')$. This graph features at
most $g(s)\cdot D\cdot n$ vertices and edges, where $g$ is a computable
function. Note that this second reduction runs in polynomial time
and also satisfies the requirements for a sub-exponential reduction.
Altogether, the theorem then follows from Corollary~\ref{cor: subexponential reductions work, contraposition}.
\end{proof}
Let us remark some corollaries of the reduction shown above.
\begin{rem}
\label{rem: Ensure bipartiteness}If the source instance $G$ for
$\coeff p$ has maximum degree $\Delta$, then the reduction images
$T(G')$ obtained above on parameter $\epsilon$ feature maximum degree
$c\Delta$ for a constant $c=c(\epsilon)$. By suitable choice of
$\mathcal{H}$, we can also ensure other properties on $T(G)$:

\begin{itemize}
\item If $G$ is bipartite and all edge-gadgets $(H,u,v)\in\mathcal{H}$
can be $2$-colored such that $u$ and $v$ receive different colors,
then $T(G')$ is bipartite as well. This can be verified, e.g., for
Example~\ref{exa: gadgets-perfmatch}.
\item If $G$ is planar and all edge-gadgets $(H,u,v)\in\mathcal{H}$ admit
planar drawings with $u$ and $v$ on their outer faces, then $T(G')$
is planar as well.
\end{itemize}
\end{rem}
To conclude this subsection, we list several possible generalizations
of Definition~\ref{def: Weight simulation} and Theorem~\ref{thm: Block interpolation}
that we chose not to include in our formulation.
\begin{enumerate}
\item We may extend Definition~\ref{def: Weight simulation} to incorporate
weight simulations that proceed non-locally, that is, in a less obvious
way than by inserting local gadgets at edges. This route was taken in \cite{DBLP:conf/iwpec/BrandDR16} since the conference version of this article was published.
\item In Lemma~\ref{lem: Weighted block-wise interpolation}, we do not
need to evaluate $\pmult p$ on a grid $W^{t}$ for a \emph{fixed}
coordinate set $W$. Instead, we could as well interpolate on a grid
$W_{1}\times\ldots\times W_{t}$, provided that each $W_{i}$ is large
enough and that the weights do not depend on the size of $G$.
\item We may also allow $2^{o(m)}$ time for the computation of the factor
function $F(G_{w},\xi)$. Rather than allowing only a multiplicative
factor, we could also allow an arbitrary function to be computed from
$p(T(G);\xi)$ and the input.
\end{enumerate}

\section{\label{sec:Applications}Applications of the Framework}

In the following subsections, we apply Theorem~\ref{thm: Block interpolation}
to obtain tight lower bounds for concrete counting problems, including the
unweighted permanent in Section~\ref{subsec:Unweighted-permanent},
the matching and independent set polynomials in Section~\ref{subsec:Matching-and-independent},
and the Tutte polynomial in Section~\ref{subsec:Tutte-polynomial}.

\subsection{\label{subsec:Unweighted-permanent}The Unweighted Permanent}

As mentioned in the introduction, it was shown in \cite{DBLP:journals/talg/DellHMTW14}
that, unless $\sharpETH$ fails, the problem $\evalfix{\{-1,1\}}{\perm}$
on graphs with $n$ vertices and $O(n)$ edges admits no algorithm
with running time $2^{o(n)}$. For convenience, we recall that this
is the problem of evaluating the permanent on graphs with edge-weights
$+1$ and $-1$. In the same paper, it was also shown that an algorithm
for the \emph{unweighted} permanent on such graphs would falsify $\rETH$,
the randomized version of $\ETH$. We improve upon this by showing
that it is sufficient to assume $\sharpETH$, which is a priori weaker
than $\ETH$ and also constitutes a more natural assumption for lower
bounds on counting problems.
\begin{thm*}
[Restatement of Theorem \ref{thm: Main PM}] Assuming $\sharpETH$,
there are constants $\epsilon,C>0$ such that the problem $\evalfix 1{\perm}$
of counting unweighted perfect matchings in bipartite graphs cannot
be solved in time $O(2^{\epsilon n})$ on graphs with $n$ vertices
and at most $Cn$ edges.
\end{thm*}
\begin{proof}
In the following, we invoke Theorem~\ref{thm: Block interpolation}
to show 
\[
\evalfix{\{-1,1\}}{\perm}\leqSERF\evalfix 1{\perm}.
\]

Let $G$ be a graph with edge-weights from $\{-1,1\}$ and let $E_{-1}(G)$
denote the set of edges with weight $-1$ in $G$. Define a sieve
function and weight selector 
\begin{eqnarray*}
\chi(G,A) & = & [A\in\PM[G]],\\
\omega(G,A) & = & A\cap E_{-1}(G),
\end{eqnarray*}
and observe that these induce a univariate graph polynomial $p=p_{\chi,\omega}$
with 
\[
p(G;-1)=\perm(G).
\]
Since knowing all coefficients of $p(G)$ clearly allows to evaluate
$p(G;-1)$, we obtain from \cite[Thm.~1.3]{DBLP:journals/talg/DellHMTW14}
that there are constants $\delta,D>0$ such that $\coeff p$ cannot
be solved in time $O(2^{\delta n})$ on $n$-vertex graphs with $Dn$
edges, unless $\sharpETH$ fails. Hence the coefficient hardness condition
of Theorem~\ref{thm: Block interpolation} is satisfied.\footnote{In fact, the authors of \cite{DBLP:journals/talg/DellHMTW14} state
their lower bound as ruling out $2^{o(n)}$ time algorithms for $\coeff p$
on graphs with $n$ vertices and $O(n)$ edges. This is however only
to simplify the presentation of their result. Their reduction from
counting satisfying assignments for $3$-CNFs to $\evalfix{\{-1,1\}}{\perm}$
is in fact a $\leqSERF$ reduction and hence also supports the stronger
claim needed for the coefficient hardness condition of Theorem~\ref{thm: Block interpolation}.}

To check the weight simulation condition, recall the pair $(\mathcal{H},F)$
from Example~\ref{exa: gadgets-perfmatch} that allows to reduce
$\evalfix{\mathbb{N}}{\pmult p}$ to $\evalfix 1p$. 
Hence all prerequisites for Theorem~\ref{thm: Block interpolation} are fulfilled,
so counting perfect matchings in unweighted graphs has the desired lower bound.
By Remark~\ref{rem: Ensure bipartiteness},
the reduction images $T(G)$ constructed by the theorem
are bipartite as well. This proves the theorem.
\end{proof}
We collect a series of corollaries for other counting problems from
this theorem. Let $L(G)$ denote the \emph{line graph} of a graph
$G=(V,E)$: This graph has vertex set $E$, and $e,e'\in E$ are defined
to be adjacent in $L(G)$ iff $e\cap e'\neq\emptyset$. A graph is a
line graph if it is the line graph of some graph.
\begin{cor}
\label{cor: Lower bounds from permanent}Assuming $\sharpETH$, there
is a constant $\epsilon>0$ for each of the following problems such
that no $O(2^{\epsilon n})$ time algorithm solves the problem:

\begin{enumerate}
\item $\evalfix 1{\perm}$, that is, counting perfect matchings in bipartite
graphs, in graphs with $n$ vertices and maximum degree $3$.
\item Counting maximum independent sets (or minimum vertex covers), even
in line graphs with $n$ vertices and maximum degree $4$.
\item Counting minimum-weight satisfying assignments to monotone $2$-CNF
formulas on $n$ variables, even if every variable appears in at most
$4$ clauses.
\end{enumerate}
\end{cor}
\begin{proof}
For the first statement, we use a known reduction from the permanent on
general bipartite graphs to bipartite graphs of maximum degree $3$,
shown in \cite[Theorem 6.2]{DBLP:journals/tcs/DagumL92}. This reduction
maps graphs with $n$ vertices and $m$ edges to graphs with $O(n+m)$
vertices and edges.

For the second statement, if $G$ has $m$ edges and maximum degree
$\Delta=\Delta(G)$, then $L(G)$ has $m$ vertices and maximum degree
$2(\Delta-1)$. The set $\PM[G]$ corresponds bijectively to the independent
sets of size $n/2$ in $L(G)$, which are the maximum independent
sets in $L(G)$, unless $G$ has no perfect matching, which we can
test efficiently. The maximum independent sets in turn stand in bijection
with the minimum vertex covers of $L(G)$ via complementation. We
thus obtain the statement by reduction from $\evalfix 1{\perm}$ on
graphs of maximum degree $3$.

For the third statement, observe that the minimum vertex covers of
a graph $H=(V,E)$ correspond bijectively to the minimum-weight satisfying
assignments of an associated monotone $2$-CNF formula $\varphi$:
To obtain this formula, create a variable $x_{v}$ for each $v\in V$ and a clause $(x_{u}\vee x_{v})$
for each $uv\in E$. This is a standard reduction, noted also in \cite[Proposition 2.1]{DBLP:journals/siamcomp/Vadhan01}.
\end{proof}

\subsection{\label{subsec:Matching-and-independent}The Matching and Independent
Set Polynomials}

We prove Theorem~\ref{thm: Main Match}, a tight lower bound for the evaluation problem of the matching polynomial $\evalfix{\xi}{\mu}$ at fixed $\xi\in\mathbb{Q}$, by invoking Theorem~\ref{thm: Block interpolation}.
As described in the introduction, the perfect matchings of $G$ are
counted by the coefficient of $x^{0}$ in $\mu(G)$, so $\coeff{\mu}$
and $\evalfix 0{\mu}$ have the same lower bound as $\evalfix 1{\perm}$.
This settles the requirement of coefficient hardness in Theorem~\ref{thm: Block interpolation}.
In the following, we analyze the example for an interpolation-based reduction from the introduction (where we reduced counting perfect matchings to counting matchings)
to show that $\mu$ also admits weight simulation.
\begin{lem}
\label{lem: mu MatchSum weights}Let $\mathcal{H}=(H_{i})_{i\in\mathbb{N}}$
be the following family of vertex-gadgets, where $H_{i}$ contains
one attachment vertex $v$, adjacent to an independent set of $i$
vertices. 

\begin{center}
\includegraphics[width=0.18\textwidth]{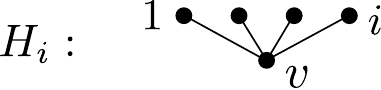}
\par\end{center}

Consider $\xi\in\mathbb{Q}\setminus\{0\}$ to be fixed. Let $W=(w_{t})_{t\in\mathbb{N}}$
with $w_{t}=1+\frac{t}{\xi^{2}}$ for $t\in\mathbb{N}$. Given a graph
$G$ with vertex-weights from $W$, let $a_{t}$ for $t\in\mathbb{N}$
denote the number of vertices of weight $w_{t}$ in $G$. We define
\[
F(G,\xi)=\prod_{t\in\mathbb{N}}\xi^{t\cdot a_{t}}.
\]
Then $(\mathcal{H},F)$ allows to reduce $\evalfix W{\pmult{\mu}}$
to $\evalfix{\xi}{\mu}$.
\end{lem}
\begin{proof}
Recall the graph transformation $T(G)$ for vertex-weighted graphs
from Definition~\ref{def: Weight simulation} and Remark~\ref{rem:vertex-weights}:
At every vertex of weight $w_{t}$, for $t\in\mathbb{N}$, we attach
the gadget $H_{t}$. To show that $(\mathcal{H},F)$ indeed satisfies
Definition~\ref{def: Weight simulation}, we need to show that 
\begin{equation}
\mu(T(G),\xi)=F(G,\xi)\cdot\pmult{\mu}(G).\label{eq: mu =00003D MatchSum}
\end{equation}

To see this, observe that every matching $M$ in $G$ can be extended
locally at vertices $v\in V(G)$ by edges of vertex-gadgets to obtain
a matching in $T(G)$. Let $v\in V(G)$ be a vertex of weight $w_{t}$
for $t\in\mathbb{N}$. If $v\notin\usat(G,M)$, then $M$ cannot be
extended at the vertex $v$, and $v$ incurs the factor $\xi^{t}$
in $\mu(T(G))$. If $v\in\usat(G,M)$, then we can choose not to extend
$v$, or we may choose to extend by one of the $t$ edges of $H_{t}$,
so we obtain a factor of 
\[
\xi^{t}+t\xi^{t-2}=\xi^{t}(1+\frac{t}{\xi^{2}}).
\]
at the vertex $v$. This establishes (\ref{eq: mu =00003D MatchSum}), and hence
the lemma.
\end{proof}
The desired theorem follows.
\begin{thm*}
[Restatement of Theorem \ref{thm: Main Match}] If $\sharpETH$ holds,
then for each $\xi\in\mathbb{Q}$, there are constants $\epsilon,C>0$
such that $\evalfix{\xi}{\mu}$ cannot be solved in time $O(2^{\epsilon n})$
on graphs with $n$ vertices and maximum degree $C$. With $\xi=1$,
this holds especially for $\evalfix 1{\mu}$, which amounts to counting
matchings.
\end{thm*}
\begin{proof}
By Corollary~\ref{cor: Lower bounds from permanent}, there is a
constant $\delta>0$ such that $\coeff{\mu}$ cannot be solved in
time $O(2^{\delta n})$ unless $\sharpETH$ fails, even on graphs
with maximum degree $3$. This settles the requirement of coefficient
hardness in Theorem~\ref{thm: Block interpolation}, even on graphs
of maximum degree $3$.

In Lemma~\ref{lem: mu MatchSum weights}, we have seen that $\mu$
admits weight simulations. By Remark~\ref{rem: Ensure bipartiteness}
and the reduction from $\coeff{\mu}$ on graphs of maximum degree
$3$, the queries issued by Theorem~\ref{thm: Block interpolation}
have maximum degree $O(1)$, which implies the degree bound in the
theorem.
\end{proof}
\begin{rem}
\label{rem: match eth also complex}For later use, note that
the same proof yields the same lower bound for $\mu(G;\xi)$ even
when $\xi\in\mathbb{C}$ is a complex number with $\xi=\sqrt{c}$
for $c\in\mathbb{Q}$. To this end, note that we may assume that $G$
has an even number of vertices (for instance, by adding an isolated
vertex and dividing $\mu(G;\xi)$ by $\xi$). Then we can compute
$\mu(G;\xi)$ over the rational numbers: Every matching in $G$ has
an even number of unmatched vertices, and thus only even powers of
$\xi$ appear in $\mu(G;\xi)$.
\end{rem}
As in Corollary~\ref{cor: Lower bounds from permanent}, we can easily
obtain corollaries for the independent set polynomial $I$ and for counting
monotone $2$-SAT, improving upon \cite{DBLP:conf/iwpec/Hoffmann10,DBLP:journals/talg/DellHMTW14}. 
\begin{cor}
Assuming $\sharpETH$, there is a constant $\epsilon>0$ for each
of the following problems such that no $O(2^{\epsilon n})$ time algorithm
solves the problem:

\begin{enumerate}
\item $\evalfix{\xi}I$ on line graphs of maximum degree $O(1)$, for $\xi\in\mathbb{Q}\setminus\{0\}$,
especially at $\xi=1$, which amounts to counting independent sets
(or vertex covers).
\item Counting satisfying assignments to monotone 2-CNF formulas, even if
every variable appears in at most $O(1)$ clauses.
\end{enumerate}
\end{cor}
\begin{proof}
We first prove the first statement: If $G$ has $n$ vertices and $m$ edges,
and satisfies $\Delta(G)=O(1)$, then $L(G)$ has $m$ vertices and
$\Delta(L(G))=O(1)$. For every matching $M\in\M[G]$, we have
\[
2|M|+|\usat(M)|=n,
\]
and since matchings of $G$ stand in bijection with independent sets of $L(G)$,
\[
\mu(G;\xi)=\sum_{M\in\M[G]}\xi^{|\usat(M)|}=\xi^{n}\cdot\sum_{M\in\M[G]}\xi^{-2|M|}=\xi^{n}\cdot I(L(G);\xi^{-2}).
\]

Hence, for fixed $\xi\neq0$, an algorithm for $\evalfix{\xi}I$ on
line graphs implies one for $\evalfix{\xi'}{\mu}$ on general graphs
with $\xi'=\sqrt{\xi^{-1}}$, but this is ruled out by Theorem~\ref{thm: Main Match}
and Remark~\ref{rem: match eth also complex}.

For the second statement, recall that independent sets and vertex
covers stand in bijection. We then reduce from counting vertex covers
as in Corollary~\ref{cor: Lower bounds from permanent}. 
\end{proof}

\subsection{\label{subsec:Tutte-polynomial}The Tutte Polynomial}

To apply the block-interpolation framework to the Tutte polynomial,
we use univariate restrictions of $Z$, as discussed in Example~\ref{exa: tutte-reformulated}.
Let $Z_{q}=Z(q,\cdot)$ for fixed $q\in\mathbb{Q}\setminus\{0\}$.
As in \cite{DBLP:journals/talg/DellHMTW14}, for $q=0$, we instead
consider the polynomial 
\[
Z_{0}(G;w)=\sum_{A\subseteq E(G)}0^{k(G,A)-k(G,E)}w^{|A|}.
\]
Note that terms corresponding to $A\subseteq E(G)$ with $k(G,A)\neq k(G,E)$
vanish in $Z_{0}(G;w)$. We will use $Z_{0}$ to prove lower bounds
for the Tutte polynomial on the line defined by $x=1$: If $(x,y)\in\mathbb{Q}^{2}$
satisfies $x=1$, we can write $q=0$ and $w=y-1$ and obtain 
\begin{eqnarray}
T(G;x,y) & = & (x-1)^{-k(G,E)}\cdot(y-1)^{-|V(G)|}\cdot Z(G;q,w)\nonumber \\
 & = & w^{-|V(G)|}\cdot Z_{0}(G;w).\label{eq: Z to Z0}
\end{eqnarray}

Using Theorem~\ref{thm: Block interpolation}, we then prove lower
bounds for $\evalfix w{Z_{q}}$ at fixed $q,w\in\mathbb{Q}$. As in
the previous examples, we first require a lower bound for $\coeff{Z_{q}}$,
which we adapt from \cite{DBLP:journals/talg/DellHMTW14}.
\begin{lem}
\label{lem: tutte-coeff}\cite[Propositions~4.1~and~4.3]{DBLP:journals/talg/DellHMTW14}
Assuming $\sharpETH$, for each $q\in\mathbb{Q}\setminus\{1\}$, there
are constants $\epsilon,C>0$ such that the problem $\coeff{Z_{q}}$
cannot be solved in time $O(2^{\epsilon n})$ on $n$-vertex graphs
with $Cn$ edges.\footnote{In \cite{DBLP:journals/talg/DellHMTW14}, this result is stated as
$\coeff{Z_{q}}$ not having $2^{o(m)}$ time algorithms for $q\in\mathbb{Q}\setminus\{1\}$.
However, the paper actually shows the slightly stronger claim in the
statement of the lemma.}
\end{lem}
\begin{rem}
In fact, we could also use block interpolation to simplify this result
from \cite{DBLP:journals/talg/DellHMTW14} by performing an interpolation
step that needed to be circumvented by the authors with some tricks.
However, since Lemma~\ref{lem: tutte-coeff} was already shown in
\cite{DBLP:journals/talg/DellHMTW14}, we omit the self-contained
proof that would still require some arguments which are very specific
to the Tutte polynomial.
\end{rem}
Note that the case $q=1$ is left uncovered by this lemma, and we
consequently cannot prove lower bounds at $q=1$, where $\coeff{Z_{1}}$
in fact becomes polynomial-time solvable. 

In \cite{DBLP:journals/talg/DellHMTW14}, the problem $\coeff{Z_{q}}$
with $q\neq1$ is reduced to unweighted evaluation via \emph{Theta
graphs} and \emph{wumps}, families of edge-gadgets that incur only
$O(\log^{c}n)$ blowup. This economical (but still\emph{ }not constant)
factor however requires a quite involved analysis. Using block interpolation,
we can instead use mere paths, and hence perform \emph{stretching},
a classical weight simulation technique for the Tutte polynomial \cite{JVW90,DBLP:journals/talg/DellHMTW14,DBLP:journals/siamcomp/GoldbergJ14}.
In the following, please recall $\pmult Z_{q}$, as defined by Example~\ref{exa: tutte-reformulated}
and (\ref{eq: multivar-p}).
\begin{lem}
\label{lem: tutte-weights}For $k\in\mathbb{N}$, let $P_{k}$ denote
the path on $k$ edges with distinguished start/end vertices $u,v\in V(P_{k})$
and let $\mathcal{P}=(P_{1},P_{2},\ldots)$. Let $w,q\in\mathbb{Q}$
be fixed with $w\neq0$ and $q\notin\{1,-w,-2w\}$. Then there is
an infinite recursively enumerable sequence of pairwise distinct weights
$W$ and a factor function $F$ such that $(\mathcal{P},F)$ allows
to reduce $\evalfix W{\pmult Z_{q}}$ to $\evalfix w{Z_{q}}$.
\end{lem}
\begin{proof}
We have to distinguish whether $q=0$ or $q\neq0$ holds, and we obtain
different weights and factor functions in the different cases.

If $q=0$, we define $W=(w_{k})_{k\in\mathbb{N}}$ with the pairwise
distinct weights $w_{k}=\frac{w}{k}$ for integers $k\in\mathbb{N}$. Given
a graph $G$ with edge-weights from $W$, let $a_{k}(G)$ for $k\in\mathbb{N}$
denote the number of edges in $G$ with weight $w_{k}$, and define
\[
F(G)=\prod_{k\in\mathbb{N}}(kw^{k-1})^{a_{k}(G)}.
\]
Then $(\mathcal{P},F)$ allows to reduce $\evalfix W{\pmult Z_{0}}$
to $\evalfix w{Z_{0}}$, see for instance \cite[Corollary~6.7]{DBLP:journals/talg/DellHMTW14}
and \cite{GJ08Tutte}.

If $q\neq0$, then the family of paths realizes different weights
and requires a different factor function. Define $W=(w_{k})_{k\in\mathbb{N}}$
with 
\[
w_{k}=\frac{q}{(1+\frac{q}{w})^{k}-1}
\]
and observe that these weights are pairwise distinct provided that
$1+\frac{q}{w}\notin\{-1,0,1\}$, which holds by $q\neq0$ and the
prerequisites of the proposition. Given a graph $G$ with edge-weights
from $W$, let $a_{k}(G)$ for $k\in\mathbb{N}$ denote the number
of edges in $G$ with weight $w_{k}$ and define 
\[
F(G)=q^{-|E(G)|}\prod_{k\in\mathbb{N}}((q+w)^{k}-w^{k})^{a_{k}(G)}.
\]
It is shown \cite[Lemma~6.2]{DBLP:journals/talg/DellHMTW14} and
\cite[Prop.~2.2 and 2.3]{Sokal2005} that $(\mathcal{P},F)$
allows to reduce $\pmult Z_{q}$ on $W$ to $Z_{q}(w)$.
\end{proof}
By combining Lemma~\ref{lem: tutte-coeff} for the coefficient hardness
and Lemma~\ref{lem: tutte-weights} for weight simulations, we can
then invoke Theorem~\ref{thm: Block interpolation} and obtain:
\begin{lem}
\label{lem: hardness-Q}Let $w\neq0$ and $q\notin\{1,-w,-2w\}$.
Assuming $\sharpETH$, there are constants $\epsilon,C>0$ such that
the problem $\evalfix w{Z_{q}}$ admits no $O(2^{\epsilon n})$ time
algorithm on graphs with $n$ vertices and at most $Cn$ edges.
\end{lem}
From this lemma, we can derive the hardness of the problem $Z(q,w)$ at most points
with $q\notin\{0,1\}$ analogously as in \cite[Proposition~6.4]{DBLP:journals/talg/DellHMTW14}.
\begin{lem}
\label{lem: Q-all-hardness}Let $(q,w)\in\mathbb{Q}^{2}\setminus\{(4,-2),(2,-1),(2,-2)\}$
with $q\notin\{0,1\}$ and $w\neq0$. Assuming $\sharpETH$, there
are constants $\epsilon,C>0$ such that the problem $\evalfix w{Z_{q}}$
admits no $O(2^{\epsilon n})$ time algorithm on graphs with $n$
vertices and at most $Cn$ edges.
\end{lem}
\begin{proof}
By Lemma~\ref{lem: hardness-Q}, the claim must only be shown if 
$q\in\{-w,-2w\}$ holds in addition to the prerequisites of Lemma~\ref{lem: Q-all-hardness}. As in \cite[Proposition~6.4]{DBLP:journals/talg/DellHMTW14},
we then use the operations of thickening and stretching to reduce
the problem $\evalfix{w'}{Z_{q}}$ for some $w'$ with $q\notin\{-w',-2w'\}$
to $\evalfix w{Z_{q}}$. The hardness of $\evalfix w{Z_{q}}$ then
follows from Lemma~\ref{lem: hardness-Q}.

To proceed this way, let $G_{k}$ be the graph obtained from $G$ by replacing
each edge with $k$ parallel edges, followed by subdividing each edge
once. Then there exists a number $w_{k}$, depending on $q$, $w$,
and $k$, such that $Z(G;q,w_{k})$ can be computed in polynomial
time from the value $Z(G_{k};q,w)$, as shown in \cite[Proposition~6.4]{DBLP:journals/talg/DellHMTW14}.
The same reference shows that, if the prerequisites of the lemma are
satisfied, a suitable value $k=k(q,w)$ can be chosen such that $q\notin\{-w_{k},-2w_{k}\}$.
Since $q,w$ are fixed, we have $k=O(1)$, and the graph $G_{k}$
hence has $O(|V(G)|+|E(G)|)$ vertices and edges. This proves the
claim.
\end{proof}
By the substitution (\ref{eq: tutte-relations}) that maps $Z(\cdot,\cdot)$
to the classical parameterization $T(\cdot,\cdot)$ of the Tutte polynomial,
we can rephrase this result in terms of $T$.
\begin{thm*}
[Restatement of Theorem \ref{thm: Main Tutte}] Assuming $\sharpETH$,
there are constants $\epsilon,C>0$ such that the Tutte polynomial
$T(x,y)$ cannot be evaluated in time $O(2^{\epsilon n})$ on graphs
with $n$ vertices and at most $Cn$ edges, provided that $y\neq1$,
and $(x,y)\notin\{(-1,-1),(0,-1),(-1,0)\}$, and $(x-1)(y-1)\neq1$.

\begin{proof}
Using (\ref{eq: tutte-relations}), computing $Z(G;q,w)$ is equivalent
to computing $T(G;x,y)$ with $x=\frac{q}{w}+1$ and $y=w+1$, provided
that $q\neq0$. We use this to rephrase the evaluations $Z(q,w)$
for $(q,w)\in\mathbb{Q}^{2}$ that are not shown to be hard by Lemma~\ref{lem: Q-all-hardness}
in terms of $T(x,y)$.

\begin{enumerate}
\item If $w=0$, then $y=1$.
\item If $(q,w) \in \{(4,-2),(2,-1),(2,-2)\}$, then $(x,y)\in\{(-1,-1),(-1,0),(0,-1)\}$.
\item If $q=1$, then $(x-1)(y-1)=1$.
\end{enumerate}
Hence, Lemma~\ref{lem: Q-all-hardness} shows a tight lower bound
for all points $(x,y)\in\mathbb{Q}^{2}$ relevant for the theorem
that satisfy $q=(x-1)(y-1)\neq0$. We then consider those points with
$q=0$. Since we may assume $y\neq1$, only points $(x,y)$ with $x=1$
and $y\neq1$ are left open. In this case, we invoke Lemma~\ref{lem: hardness-Q}
with $q=0$ and $w=y-1$. Using (\ref{eq: Z to Z0}), we then obtain
$T(G;1,y)=w^{-|V(G)|}\cdot Z_{0}(G;w)$. Since $w\neq0$ and $\evalfix w{Z_{0}}$
admits a tight lower bound under $\sharpETH$ by Lemma~\ref{lem: hardness-Q},
the theorem follows.
\end{proof}
\end{thm*}
If either of the last two conditions of Theorem~\ref{thm: Main Tutte}
does not hold, then the evaluation of the Tutte polynomial
is known to admit a polynomial-time algorithm. The $\sharpP$-hard
points on the line given by $y=1$ are however not covered by Theorem~\ref{thm: Main Tutte},
and they actually do not fit into the block interpolation framework as defined in this paper.
Nevertheless, as discussed earlier, this line was settled recently \cite{DBLP:conf/iwpec/BrandDR16} by extending the block interpolation framework to a setting where gadgets are not required to be placed locally at vertices.

\paragraph*{Acknowledgements:}

The author thanks Holger Dell and the anonymous reviewers for providing
very helpful comments that improved many aspects of this paper.

\bibliographystyle{plain}
\bibliography{references-clean}

\end{document}